\documentclass[11pt]{article}

\usepackage{palatino}
\usepackage{mathpazo}
\usepackage{braket}
\usepackage{amsfonts}
\usepackage{amssymb}
\usepackage{amsmath}
\usepackage{latexsym}
\usepackage{amsthm}
\usepackage[usenames]{color}
\usepackage{hyperref}
\usepackage{optidef}
\usepackage{cleveref}
\usepackage{relsize}
\usepackage{graphicx}
\usepackage{tikz}
\usepackage{authblk}

\usepackage[normalem]{ulem}

\theoremstyle{definition}

\hypersetup{pdfpagemode=UseNone}

\newtheorem{protocol}{Protocol} 
\newtheorem{reduction}{Reduction}  
\newtheorem{theorem}{Theorem}[section]
\newtheorem{lemma}[theorem]{Lemma}

\newtheorem{prop}[theorem]{Proposition}

\newtheorem{corollary}[theorem]{Corollary}
\theoremstyle{definition}
\newtheorem{definition}[theorem]{Definition}
\newtheorem{remark}[theorem]{Remark}

\usepackage{color,graphicx}

\newcommand{\A}{\mathcal{A}}
\newcommand{\B}{\mathcal{B}}

\newcommand{\kb}[1]{\ket{#1}\bra{#1}} 
\newcommand{\Tr}{\mathrm{Tr}}

\begin{document}

\title{\bf Quantum protocols for Rabin oblivious transfer}
\author[3]{Erika Andersson\footnote{\href{mailto:E.Andersson@hw.ac.uk}{\texttt{E.Andersson@hw.ac.uk}}}}
\author[1]{Akshay Bansal\footnote{\href{mailto:akshaybansal14@gmail.com}{\texttt{akshaybansal14@gmail.com}}}}
\author[3]{James T. Peat\footnote{\href{mailto:jtp2000@hw.ac.uk}{\texttt{jtp2000@hw.ac.uk}}}}
\author[1]{Jamie Sikora\footnote{\href{mailto:sikora@vt.edu}{\texttt{sikora@vt.edu}}}}
\author[2]{Jiawei Wu\footnote{\href{mailto:constchar0212@gmail.com}{\texttt{constchar0212@gmail.com}}}}

\affil[1]{Department of Computer Science, Virginia Polytechnic Institute and State University, Blacksburg, VA, USA}
\affil[2]{Centre for Quantum Technologies, National University of Singapore, Singapore}
\affil[3]{SUPA, Institute of Photonics and Quantum Sciences, School of Engineering and Physical Sciences, Heriot-Watt University, Edinburgh, Scotland, UK}

\date{\today}

\maketitle

\begin{abstract}
Rabin oblivious transfer is the cryptographic task where Alice wishes to receive a bit from Bob but it may get lost with probability $1/2$. 
In this work, we provide protocol designs which yield quantum protocols with improved security. 
Moreover, we provide a constant lower bound on any quantum protocol for Rabin oblivious transfer. 
To quantify the security of this task with asymmetric cheating definitions, we introduce the notion of cheating advantage which may be of independent interest in the study of other asymmetric cryptographic primitives. 
\end{abstract}

\section{Introduction}\label{sec:introduction}
Cryptographic tasks where an adversary remains oblivious to the input or the data of other parties are popularly studied under the banner of \emph{oblivious transfer}.
Several variants of oblivious transfer tasks exist in the literature today with varying goals and security definitions.
A commonly studied variant of oblivious transfer, known as 1-out-of-2 oblivious transfer (OT), is a cryptographic primitive involving two parties, Alice and Bob, where Bob has two bits of information, and Alice wishes to learn one of them based on her choice of an index bit.
The security of 
OT protocols is usually a two-fold evaluation where one needs to protect honest Bob from a cheating Alice who may wish to obtain more information that the protocol allows (perhaps by learning both bits), meanwhile Bob may wish to learn Alice's input, that is, which bit she wishes to learn (see~\Cref{def:1-out-of-2 OT} for a formal definition of 1-out-of-2 OT and its underlying security notions).

Oblivious transfer, its variants, and other simple cryptographic tasks, are referred to as primitives as they could be used as building blocks for other useful tasks such as coin flipping~\cite{kitaev2002quantum, mochon2007quantum, arora2021analytic}, bit commitment~\cite{chailloux2011}, oblivious circuit evaluation~\cite{bennett1991practical,crepeau1994quantum}, and even general multiparty computation~\cite{kilian1988founding,keller2016mascot}.

In this work, we also consider the task of coin flipping where Alice and Bob wish to flip a coin but have a different preferred choices for the outcome. 
It has been shown that quantum protocols exist for coin flipping up to arbitrary levels of security~\cite{mochon2007quantum} with explicit constructions given in~\cite{arora2021analytic,arora2024protocolsquantumweakcoin}. 
This has been generalized to construct secure $z$-balanced coin flipping in~\cite{chailloux2009optimal} where the two parties output heads with probability $z$ and tails with the remaining probability of $(1-z)$ securely. 
Moreover, optimal protocols for the \emph{semi-honest} version of 1-out-of-2 OT were developed in~\cite{chailloux2013optimal}, which also provides an the best known lower bound on the security of 1-out-of-2 oblivious transfer as a corollary.  

The focus of this work is to examine quantum protocols for a variant of oblivious transfer known as Rabin oblivious transfer.  
It has been shown that Rabin oblivious transfer is equivalent to 1-out-of-2 oblivious transfer~\cite{Crepeau1988equivalence}, and that 1-out-of-2 oblivious transfer is universal~\cite{kilian1988founding}. 
However, there are only a few papers that studied Rabin oblivious transfer in a quantum setting~\cite{bansal2023breaking,stroh2024quantum} where the aim is to study its information-theoretic security. 
The goal of our work is to provide improved bounds on its overall security by finding better protocols as well as lower bounds on how secure they can be.

\subsection{Rabin oblivious transfer}\label{sub:definitionRabinOT}

In this work, we consider a variant of oblivious transfer, known as \emph{Rabin oblivious transfer}, denoted by ROT. 
This task was first proposed by Rabin in~\cite{rabin1979digitalized,Rabin05}, whence the name. 
The task was later revisited in~\cite{fischer1996secure} to give a protocol that guarantees security under certain computational hardness assumptions. 
ROT is the cryptographic task which accomplishes the following:   
\begin{itemize}
    \item At the beginning of the protocol, Bob outputs a uniformly random bit $y \in \{ 0, 1 \}$. 
    \item At the end of the protocol, Alice receives $y$ with probability $1/2$, else the bit is lost. 
\end{itemize}

The security of ROT protocols is defined in a two-fold manner. 
Dishonest Alice wishes to maximize her chances of learning the bit $y$ while dishonest Bob wishes to maximize his chances of successfully guessing whether Alice \emph{asserts} the receipt of the data or the data loss event.
Note that if Bob cheats, Alice might think that she learned the bit $y$, but the learned bit could very well be different from $y$. 
Thus, a cheating Bob only cares if Alice \emph{thinks} she correctly learned $y$. 
Note that in contrast to the analysis described in~\cite{fischer1996secure}, we work under the regime of information-theoretic security where the adversaries are not bounded under any computational assumptions.  

We now provide a more formal definition. 
\begin{definition}[Rabin OT] Rabin OT (ROT) is the cryptographic task between two parties (Alice and Bob) such that:
    \begin{itemize}
    \item Bob chooses $y \in \{0,1\}$ uniformly at random.
    \item At the end of the protocol, Alice outputs a bit $g \in \{ 0, 1 \}$ with probability $1/2$ (which corresponds to her asserting that she learned the bit $y$) or NULL with probability $1/2$ (which corresponds to her losing the bit). 
\end{itemize}
\end{definition}

The completeness and security notions for an ROT protocol are below. 
\begin{itemize}
    \item \emph{Completeness:} If both Alice and Bob are honest, neither party aborts. Alice outputs $g = y$ with probability $1/2$ or NULL with probability $1/2$.
    \item \emph{Cheating Alice:} If Bob is honest, then dishonest Alice's cheating probability is defined as
    \begin{equation*}
        P_A^* = \max \{ \Pr[\text{Alice successfully learns $y$}] \},
    \end{equation*}
    where the maximum is taken over all cheating strategies of Alice. 
Note that $P_A^* \geq 3/4$ as Alice can simply choose to follow the protocol to learn $y$ with probability $1/2$ and randomly guess $y$ in the event of NULL.
    \item \emph{Cheating Bob}:  If Alice is honest, then dishonest Bob's cheating probability is defined as 
    \begin{equation*}
        P_B^* = \max \{ \Pr[\text{Bob correctly guesses whether Alice \emph{asserts} the receipt of $y$ or NULL}] \}, 
    \end{equation*}
    where the maximum is taken over all cheating strategies of Bob. 
    More formally, if we denote Alice's assertion by bit $a$, then $a = 0$ denotes that Alice asserted the receipt of $y$, while $a = 1$ denote that she asserted the receipt of NULL. Here, $P_B^*$ could be alternatively expressed as the maximum probability with which Bob correctly learns $a$.
    
    We emphasize that Alice asserting (possibly privately) the receipt of $y$ (by outputting $g$) does not necessarily imply the receipt of the actual $y$, i.e., it could be the case that $g \neq y$, or even that $y$ does not even exist.
    Note that $P_B^* \geq 1/2$ as Bob can simply choose to follow the protocol honestly and make a uniformly random guess to whether Alice asserts $y$ or NULL. 
\end{itemize}

As the cheating notions of Alice and Bob in the ROT are asymmetric, especially with respect to the lower bounds of $3/4$ and $1/2$ on their respective cheating probabilities, we define a single security measure of how good an ROT protocol is with respect to its pair of cheating probabilities $(P_A^*, P_B^*)$. 
To achieve this, we introduce the concept of the \emph{cheating advantage} of a protocol, denoted as $\gamma$, defined as: 
\begin{equation}\label{eq:defCheatingAdvantage}
    \begin{aligned}
        \gamma & := \max \{ \gamma_A, \gamma_B \}, \\ 
        \gamma_A & := P_A^*/P_A^{ideal} \quad \text{ where } \quad P_A^{ideal} = 3/4, \quad \text{ and } \\ 
        \gamma_B & := P_B^*/P_B^{ideal} \quad \text{ where } \quad P_B^{ideal} = 1/2. 
    \end{aligned}
\end{equation}

The ultimate objective of designing ROT protocols is thus to find one which \emph{minimizes the cheating advantage}, with the hopeful goal of making $\gamma$ as close to $1$ as possible. 
(We note that $\gamma \leq 2$ since both $P_A^*$ and $P_B^*$ are no greater than $1$.) 

We remark that this notion of \emph{cheating advantage} is similar to what is done in~\cite{chailloux2013lower} for the case of $k$-out-of-$n$ (forcing) oblivious transfer which also has vastly different ideal cheating probabilities for Alice and Bob.  
We believe that this is a natural security measure that can be adopted by other primitives, even symmetric ones. 
However, since we are dealing with two cheating parameters, it could be the case that some applications would desire to minimize a different function of these two probabilities, different than the cheating advantage. 
While we do look at cheating advantage in this work, we note that our focus is on the design of protocols that minimize both the cheating probabilities, and the cheating advantage is just a way to assign a quality measure to a given protocol.

\subsection{Prior work on quantum protocols for Rabin oblivious transfer and related cryptographic tasks}\label{sub:priorWork} 

A first quantum ROT protocol with non-trivial security guarantees was given in~\cite{bansal2023breaking} where a couple different variants in the security definition were considered including the one described in this work. 
Specifically, the authors first developed a framework of stochastic switching and then later used it to combine a couple of different insecure Rabin OT protocols to yield one with better security.
This protocol yields $\gamma_A = 1.29$ and $\gamma_B = 1.87$, implying $\gamma = 1.87 < 2$ for the security definition considered in~\Cref{sub:definitionRabinOT}.

In~\cite{stroh2024quantum}, a quantum ROT protocol using two pure states is investigated, and shown to outperform classical protocols in some parameter regimes.
Also, it was previously claimed that there exists a perfectly secure quantum protocol for Rabin OT~\cite{he2006oblivious}. However, this protocol turned out to be vulnerable to cheating~\cite{peat2024cheating}.
In fact, a constant lower bound on the overall security of a variant of ROT in the setting of secure function evaluation (with additional inputs) can be shown from the general lower bound in~\cite{osbornA} and is given explicitly in the thesis~\cite[Section 4.7]{osborn2022constant}, hinting at the existence of general lower bounds for ROT.

If one makes no assumptions on the cheating powers of Alice and Bob, for example, that they are computationally bounded, the limits on the unconditional security for 1-out-of-2 OT was studied in~\cite{chailloux2013lower} which provided a constant lower bound on the worst-case security for all complete\footnote{Complete protocols are defined such that honest participants achieve the desired objective.} protocols. 
This result extends the no-go result in~\cite{lo1997insecurity} that shows that general quantum computations cannot be done securely, which includes the special case of OT.
The lower bound was later improved in~\cite{chailloux2013optimal} which showed that a variant of OT, known as \emph{semi-honest} oblivious transfer, where dishonest Alice is restricted to the set of strategies that allows her to learn the chosen bit perfectly. 

More recently, a lower bound on the security of 1-out-of-2 OT~\cite{amiri2021imperfect} has been used to deduce fundamental trade-offs between circuit privacy, data privacy and correctness for a broad family of quantum homomorphic encryption protocols~\cite{hu2023privacy}. 

In terms of finding protocols, it was shown in~\cite{chailloux2013lower} that partial security for the 1-out-of-2 OT task can be obtained, that is, the maximum probability with which Alice or Bob could cheat successfully is strictly below $1$, a strict (quantum) improvement over classical protocols where it is known that either Alice or Bob can always cheat perfectly. 
Finding optimal protocols for 1-out-of-2 OT is still an interesting, and important, open problem.

\subsection{A few remarks about the security setting considered in this work}\label{sub:remarksSecurityDef}

Here we discuss how our study of Rabin oblivious transfer fits into the bigger picture of two-party quantum cryptography. 

First and foremost, we are considering \emph{information-theoretic} security, where we do not place any restriction on a cheating party. 
For example, we do not make any assumptions on computational hardness~\cite{yao1986how,goldreich1987how}, space limitations~\cite{schaffner2007cryptographyboundedquantumstoragemodel,damgard2008bounded}, noisy hardware~\cite{wehner2008noisy,schaffner2009robust}, etc. 
While these are interesting in their own right, and can be well-motivated, we are studying this task in the most demanding setting possible, that is, when Alice and Bob are only bounded by the laws of quantum mechanics. 

Secondly, we are assuming that all communication channels are noiseless as the problem could become easier in some noisy settings. 
Consider a binary erasure channel which transmits a bit $x$ and, with probability $p$, the bit is replaced with a flag $\perp$. 
Then for $p = 1/2$, this can accomplish exactly Rabin oblivious transfer. 
Moreover, Alice and Bob cannot cheat in this setting. 
\emph{In this sense, Rabin oblivious transfer seeks to emulate a binary erasure channel where the noise is controlled by Alice or Bob, not the environment.}  
Indeed, it is possible to efficiently implement other tasks such as bit commitment and oblivious transfer in the noisy setting~\cite{crepeaunoise97}. 

Lastly, we would like to emphasize that unlike some of the previous works that studied the composability of various two-party tasks such as coin flipping~\cite{Vilasini2019,wu2024composablesecurityweakcoin}, here we only consider the standalone security of the primitives relevant to this work, including Rabin oblivious transfer, 1-out-of-2 oblivious transfer and coin flipping. Therefore, the protocols we propose in the later sections of this work may not be necessarily secure under arbitrary compositions, and we leave this research direction for future work.

\subsection{New protocols for Rabin oblivious transfer}\label{section:protocolsRabinOT}

In this section, we present novel constructions of Rabin oblivious transfer protocols that offer a strict improvement in overall security compared to those described in previous works~\cite{bansal2023breaking,stroh2024quantum}. Here, we adopt the security notion discussed in~\Cref{sub:definitionRabinOT}.

\subsubsection{A quantum protocol via weak coin flipping}\label{subsub:badProtocolBalance} 
  
One way to reduce the cheating advantages of many two-party cryptographic protocols is to balance them using \emph{weak coin flipping}. 
Suppose that you are given two protocols \emph{for the same task} $\mathcal{P}_1$ with $\mathcal{P}_2$ and suppose that Alice can cheat more in $\mathcal{P}_1$ than in $\mathcal{P}_2$ and Bob can cheat more in $\mathcal{P}_2$ than in $\mathcal{P}_1$. 
In other words, if Alice and Bob were faced with deciding on whether to use $\mathcal{P}_1$ or $\mathcal{P}_2$, then they would want opposing decisions. 
One way to deal with this is to have them \emph{flip a coin}, and if they output HEADS then use $\mathcal{P}_1$, else use  $\mathcal{P}_2$. 
This also has the potential benefit of reducing the overall cheating, as described in the following lemma. 

\begin{lemma}[Balancing Lemma] \label{lemma:balancing}
Suppose we are given two quantum ROT protocols with $\mathcal{P}_1$ with cheating advantages $(\gamma_A^{(1)}, \gamma_B^{(1)})$ and $\mathcal{P}_2$ with cheating advantages $(\gamma_A^{(2)}, \gamma_B^{(2)})$. 
If Alice prefers $\mathcal{P}_1$ (i.e., $\gamma_A^{(1)} \geq \gamma_A^{(2)}$), Bob prefers $\mathcal{P}_2$ (i.e., $\gamma_B^{(1)} \leq \gamma_B^{(2)}$), and $0 < \frac{\gamma_B^{(2)} - \gamma_A^{(2)}}{\gamma_B^{(2)} - \gamma_A^{(2)} + \gamma_A^{(1)} - \gamma_B^{(1)}} < 1$ 
then there exists a family of quantum ROT protocols approaching 
\begin{equation} 
\gamma = \gamma_A = \gamma_B = \frac{\gamma_A^{(1)} \gamma_B^{(2)} - \gamma_B^{(1)} \gamma_A^{(2)}}{\gamma_B^{(2)} - \gamma_A^{(2)} + \gamma_A^{(1)} - \gamma_B^{(1)}} \leq \min\{\gamma_1, \gamma_2\},
\end{equation} 
where $\gamma_1 = \max\{\gamma_A^{(1)}, \gamma_B^{(1)}\}$ and $\gamma_2 = \max\{\gamma_A^{(2)}, \gamma_B^{(2)}\}$.
The improvement is strict, i.e., we have $\gamma < \min\{\gamma_1, \gamma_2\}$ under the previous conditions if and only if Alice has a strict preference for $\mathcal{P}_1$ ($\gamma_A^{(1)} > \gamma_A^{(2)}$) and Bob strictly prefers $\mathcal{P}_2$ ($\gamma_B^{(1)} < \gamma_B^{(2)}$).
\end{lemma} 
\begin{remark}
    This balancing lemma above can in general be used for any two protocols for a given two-party task (coin flipping, bit commitment, etc.)~to potentially yield another protocol with improved security. 
    We note that this trick uses a quantum protocol for coin flipping for the information-theoretic security setting. 
    In other settings, one must be mindful of the coin flipping protocol security when balancing in this way.
\end{remark}

We note that if both Alice and Bob prefer the same protocol over the other, e.g., $\gamma_A^{(1)} \leq \gamma_A^{(2)}$ and $\gamma_B^{(1)} \leq \gamma_B^{(2)}$, then this balancing lemma does not offer any advantages. 
However, its power comes from the ability to combine protocols for which Alice and Bob have different preferences.

\begin{proof}[Proof Sketch]
The proof follows immediately using the task of \emph{unbalanced weak coin flipping} for which optimal quantum protocols are known~\cite{chailloux2009optimal}, based on optimal weak (balanced) coin flipping protocols~\cite{mochon2007quantum, arora2021analytic}.   
Specifically, for the $z$-unbalanced weak coin flipping task between Alice and Bob, for all choices of $\epsilon > 0$, there exist a sequence of protocols which in the limit have the following features: 
\begin{itemize} 
\item The probability that the outcome is HEADS is $z$ when Alice and Bob are honest, and they both output the same coin flip; 
\item Alice cannot force the outcome HEADS with probability greater than $z$; 
\item Bob cannot force the outcome TAILS with probability greater than $1-z$.  
\end{itemize} 

Suppose Alice and Bob use a $z$-unbalanced weak coin flipping protocol with $z = \frac{\gamma_B^{(2)} - \gamma_A^{(2)}}{\gamma_B^{(2)} - \gamma_A^{(2)} + \gamma_A^{(1)} - \gamma_B^{(1)}}$, then on observing the outcome HEADS, they perform $\mathcal{P}_1$ and on outcome TAILS then perform $\mathcal{P}_2$. 
The calculation of $\gamma_A$ or $\gamma_B$ follows directly by calculating $z\gamma_A^{(1)} + (1-z)\gamma_A^{(2)}$ or $z\gamma_B^{(1)} + (1-z)\gamma_B^{(2)}$ for $z = \frac{\gamma_B^{(2)} - \gamma_A^{(2)}}{\gamma_B^{(2)} - \gamma_A^{(2)} + \gamma_A^{(1)} - \gamma_B^{(1)}}$.
\end{proof} 

This lemma turns out to be valuable for minimizing the cheating advantage in the protocols in this paper, mostly due to the fact that Alice and Bob's cheating probabilities are rather imbalanced. 
But before that, we demonstrate that this simple act of balancing can turn bad protocols into decent protocols. 

Consider the following two \emph{bad} protocols for ROT. 

\begin{protocol}{A bad Rabin oblivious transfer protocol $\mathcal{P}_1$.}\label{protocol:firstBad}

\begin{itemize} 
\item Bob chooses $y \in \{ 0, 1 \}$ uniformly at random. 
\item Bob flips a fair coin. 
    \begin{itemize}
      \item If HEADS, Bob announces $y$ to Alice. 
      \item If TAILS, Bob announces ``NULL'' to Alice.
    \end{itemize} 
\end{itemize}
\end{protocol} 

Clearly this protocol satisfies $P_A^* = 3/4$ since Alice either learns $y$ or not, and if not she has to guess. 
Moreover, $P_B^* = 1$ as Bob knows exactly what Alice would assert. 
Thus, for this protocol we have 
\begin{equation} 
\gamma_A = \frac{3/4}{3/4} = 1, \quad \gamma_B = \frac{1}{1/2} = 2, \quad \text{ and } \quad \gamma = \max \{ \gamma_A, \gamma_B \} = 2, 
\end{equation} 
which is the worst value of $\gamma$ possible. 

\begin{protocol}{Another bad Rabin oblivious transfer protocol $\mathcal{P}_2$.}\label{protocol:secondBad}

\begin{itemize} 
\item Bob chooses $y \in \{ 0, 1 \}$ uniformly at random. 
\item Bob prepares $\ket{y}$ and sends it to Alice. 
\item Alice flips a fair coin. 
    \begin{itemize}
      \item If HEADS, Alice measures $\ket{y}$ to learn $y$. 
      \item If TAILS, Alice throws the qubit away and accepts NULL as her outcome. 
    \end{itemize}
\end{itemize}
\end{protocol} 

Clearly this protocol satisfies $P_A^* = 1$ since she can simply measure the qubit to learn it perfectly.  
Moreover, $P_B^* = 1/2$, since Alice has full control over what she asserts, and no information is sent to Bob after Alice decides to measure or not.  
Thus, for this protocol we have 
\begin{equation} 
\gamma_A = \frac{1}{3/4} = \frac{4}{3}, \quad \gamma_B = \frac{1/2}{1/2} = 1, \quad \text{ and } \quad \gamma = \max \{ \gamma_A, \gamma_B \} = \frac{4}{3},  
\end{equation} 
which is not the worst value of $\gamma$ possible, but still not that good.   

Even though protocols $\mathcal{P}_1$ and $\mathcal{P}_2$ have poor security, we can apply the Balancing Lemma to improve their overall security resulting in the following theorem.

\begin{theorem} 
There exists a quantum protocol for ROT with cheating advantage $\gamma = 5/4$.  
\end{theorem}

A few remarks are in order. 
Note that even though this protocol is easy to describe, just use optimal unbalanced weak coin flipping with the two bad ROT protocols, it is important to note that optimal unbalanced weak coin flipping protocols are very complicated. 
The only way we know how to construct such protocols is via optimal weak coin flipping~\cite{mochon2007quantum}. 
Despite much work done to simplify such protocols~\cite{aharonovwcf16,arora2021analytic}, they are still very technical. 
Moreover, it has been shown that such protocols cannot be efficient~\cite{miller2020impossibility}. 
Thus, this begs the question: 

\smallskip  
\begin{quote} 
\textit{Does there exist ROT protocols that are simple and offer decent security?} 
\end{quote} 
\smallskip 

In the following two sections, we describe how to achieve this.
 
\subsubsection{A quantum protocol via a reduction from oblivious transfer}\label{subsection:protocolFromCKS13} 

Typically, reductions between tasks are useful to establish relationships between them, which may include non-trivial security bounds~\cite{chailloux2013lower}. 
In some cases, reductions can also be used to yield protocols with desirable security. 
However, some of the fundamental security notions, such as \emph{completeness} might be compromised while reducing one task to another. 
Here, it is crucial to devise appropriate reductions for tasks where no such violations are permitted. 

We next depict a folklore (classical) reduction from OT to ROT. 
Alice and Bob communicate only classically outside the OT subroutine to perform the ROT task. 
We characterize this as a classical reduction given that the communication outside of the OT subroutine is classical. 

\begin{reduction}{Rabin oblivious transfer from 1-out-of-2-oblivious transfer.}\label{reduction:RabinOTFromOT}
    \begin{itemize} 
        \item Bob chooses data $y \in \{0,1\}$ and a dummy bit $z \in \{0,1\}$ uniformly at random.
        \item Bob chooses a permutation value $p \in \{0,1\}$ uniformly at random and assigns the bits $(x_0,x_1) = (y,z)$ if $p = 0$, or $(x_0,x_1) = (z,y)$, otherwise.
        \item \textbf{Alice and Bob perform $1$-out-of-$2$ oblivious transfer} so that Alice learns either $x_0$ or $x_1$ (see details on 1-out-of-2 OT in~\Cref{def:1-out-of-2 OT}). 
        \item Bob sends $p$ to Alice. 
        Alice now knows whether she learned $y$ or $z$.
        \item 
        If she learned $z$, then she outputs NULL. 
        Otherwise, she outputs $y$. 
    \end{itemize}
\end{reduction}

It is straightforward to observe that an ROT implementation using~\Cref{reduction:RabinOTFromOT} is \emph{complete} whenever the underlying OT protocol is \emph{complete}. 
Additionally, the reduction is tight in the sense that given a perfectly secure OT protocol the resultant ROT protocol also enjoys perfect security. 
Of course, perfectly secure OT is impossible as discussed earlier, thus this reduction does not give us perfectly secure ROT protocols. 
However, it does give partial security, as discussed below. 

Using this reduction, our idea to create ROT protocols is now clear, we can simply choose an existing (or develop a new) quantum OT protocol and slot it into~\Cref{reduction:RabinOTFromOT}. 
However, since this is a quantum protocol, it is \emph{not clear} how Alice and/or Bob may cheat with regards to composing various cheating strategies. 
That is, an optimal cheating strategy in the OT protocol may not be the optimal cheating strategy in the overall protocol.  
This optimal cheating behavior (and thus the worst case cheating probabilities) needs to be proven for this reduction, since we do not know much about the security aspects of 1-out-of 2 oblivious transfer protocols when used as a subroutine in a larger composition.  
Thus, we provide a rigorous self-contained security analysis of the Rabin OT protocols resulting from the previous reduction.

We now exhibit a complete quantum protocol for ROT, below. 

\begin{protocol}{Using~\Cref{reduction:RabinOTFromOT} with the 1-out-of-2 OT protocol from~\cite{chailloux2013lower}.}\label{protocol:reductionFromCKS}

    \begin{itemize} 
        \item Alice chooses a bit $a \in \{0,1\}$ uniformly at random and creates the two-qutrit state
            \begin{equation*}
                \ket{\phi_a} = \frac{1}{\sqrt{2}}\ket{aa} + \frac{1}{\sqrt{2}}\ket{22} \in \mathcal{A} \otimes \mathcal{B}.
            \end{equation*}
        \item Alice sends the qutrit $\mathcal{B}$ to Bob.
        \item Bob chooses data bit $y \in \{0,1\}$ and $z \in \{0,1\}$  uniformly at random.
        \item Bob chooses a permutation bit $p \in \{0,1\}$ uniformly at random and assigns $(x_0,x_1) = (y,z)$ if $p = 0$, or $(x_0,x_1) = (z,y)$, otherwise.
        \item Bob applies the unitary
        \begin{equation*}
            U_{x_0x_1} = 
            \begin{bmatrix}
                (-1)^{x_0} & 0 & 0 \\
                0 & (-1)^{x_1} & 0 \\
                0 & 0 & 1
            \end{bmatrix}
        \end{equation*}
        to $\B$ and sends it back to Alice.
        \item Alice determines the value of $x_a$ using the two-outcome measurement:
        \begin{equation*}
            \{\Pi_{0} \coloneqq \kb{\phi_a}_{\A \otimes \B}, \Pi_{1} \coloneqq \mathbb{1}_{\A \otimes \B} - \kb{\phi_a}  \}.
        \end{equation*}
        \item Bob sends $p$ to Alice.
        \item If $a = p$, Alice asserts the receipt of $y = x_a$, else she asserts NULL.
    \end{itemize}
    
\end{protocol}

The completeness of the 1-out-of-2 protocol in~\cite{chailloux2013lower} ensures that this ROT protocol is also \emph{complete}. However, it is not clear whether the optimal cheating probability for dishonest Alice where she tries to learn Bob's data bit $y$ is equivalent to Alice learning both $x_0$ and $x_1$ in the reduced 1-out-of-2 OT protocol. Although learning just these two bits ($x_0$ and $x_1$) qualifies as a valid cheating strategy, she might deploy a more clever strategy to improve her chances of successfully learning $y$, i.e., we have $P_A^* \geq P_A^{OT}$. Note that quantum communication in the reduced 1-out-of-2 OT protocol attributes to a non-trivial overall analysis as dishonest Alice could now correlate her messages with Bob's input data. 
We give an argument of Alice's possible strategies below, and provide the missing details in the appendix.

A general strategy for a cheating Alice has her creating a state of the form below 
\begin{equation*}
    \ket{\psi'} =  \alpha \ket{\psi_{\alpha}}_{\mathcal{A}'} \ket{0}_{\mathcal{B}} + \beta \ket{\psi_{\beta}}_{\mathcal{A}'} \ket{1}_{\mathcal{B}} + \gamma \ket{\psi_{\gamma}}_{\mathcal{A}'} \ket{2}_{\mathcal{B}},
\end{equation*}
where $\alpha, \beta, \gamma \geq 0$ satisfy $\alpha^2 + \beta^2 + \gamma^2 = 1$, and sends the subsystem $\mathcal{B}$ to Bob. 
Once Bob applies the unitary $U_{x_0 x_1}$ on the subsystem $\mathcal{B}$ and sends it back, Alice holds the state
\begin{equation*}
    \ket{\psi_{x_0 x_1}'} = \alpha (-1)^{x_0} \ket{\psi_{\alpha}}_{\mathcal{A}'} \ket{0}_{\mathcal{B}} + \beta (-1)^{x_1} \ket{\psi_{\beta}}_{\mathcal{A}'} \ket{1}_{\mathcal{B}} + \gamma \ket{\psi_{\gamma}}_{\mathcal{A}'} \ket{2}_{\mathcal{B}}.
\end{equation*}
It is important to observe that Alice's choice of the states $\{\ket{\psi_{\alpha}}, \ket{\psi_{\beta}}, \ket{\psi_{\gamma}}\}$ is not relevant in maximizing her chances of successfully cheating in~\Cref{protocol:reductionFromCKS} as for any choice of $\ket{\psi_{\alpha}}$, $\ket{\psi_{\beta}}$ and $\ket{\psi_{\gamma}}$, there exists a unitary $U$ such that 
\begin{equation*}
    U \ket{\psi'} = \ket{\psi},
\end{equation*}
where $\ket{\psi} = \alpha \ket{0}_{\mathcal{A}'} \ket{0}_{\mathcal{B}} + \beta \ket{0}_{\mathcal{A}'} \ket{1}_{\mathcal{B}} + \gamma \ket{0}_{\mathcal{A}'}\ket{2}_{\mathcal{B}}$.\\
Thus, Alice could restrict to creating states of the form $\ket{\psi}$, or equivalently just the states of the form 
\begin{equation*}
    \ket{\psi} =  \alpha \ket{0} + \beta \ket{1} + \gamma \ket{2} \in \mathcal{B}. 
\end{equation*} 
Subsequently, her task reduces to extracting $x_0$ or $x_1$ (depending on $p$) from the reduced state 
\begin{equation}
    \ket{\psi_{x_0 x_1}} = \alpha (-1)^{x_0} \ket{0} + \beta (-1)^{x_1} \ket{1} + \gamma \ket{2}  \in \mathcal{B}.
\end{equation}
If $p = 0$, the probability of Alice successfully guessing $y$ (or $x_0$ in this case) is given by
\begin{equation*}
    \Pr[\text{Alice successfully learns $y$} | p = 0] = \frac{1}{2} + \frac{1}{4} \Big\lVert \rho_{0}(0) - \rho_{1}(0) \Big\rVert_1,   
\end{equation*}
where $\rho_0(0) = \frac{1}{2}\kb{\psi_{00}} + \frac{1}{2}\kb{\psi_{01}}$ and $\rho_1(0) = \frac{1}{2}\kb{\psi_{10}} + \frac{1}{2}\kb{\psi_{11}}$.\\
Similarly, if $p = 1$, the probability of Alice successfully guessing $y$ (or $x_1$ in this case) is given by
\begin{equation*}
    \Pr[\text{Alice successfully learns $y$} | p = 1] = \frac{1}{2} + \frac{1}{4} \Big\lVert \rho_{0}(1) - \rho_{1}(1) \Big\rVert_1,   
\end{equation*}
where $\rho_0(1) = \frac{1}{2}\kb{\psi_{00}} + \frac{1}{2}\kb{\psi_{10}}$ and $\rho_1(1) = \frac{1}{2}\kb{\psi_{01}} + \frac{1}{2}\kb{\psi_{11}}$.\\
Optimizing the overall success probability over all choices of $\alpha$, $\beta$ and $\gamma$, we get $P_A^{*} = \cos^2(\pi/8)$ when $\alpha = \beta = 1/2$ and $\gamma = 1/\sqrt{2}$. 

We now state the overall security of this protocol and defer the rest of the proof to~\Cref{appendix:proofReductionToCKS}.

\begin{lemma}\label{lemma:ROTWithCKS}
    The ROT~\Cref{protocol:reductionFromCKS} has cheating probabilities $P_A^* = \cos^2(\pi/8)$ and $P_B^* = 3/4$ implying $\gamma_A = 1.138$ and $\gamma_B = 1.5$, resulting in a cheating advantage of $\gamma = 1.5$. 
\end{lemma}

By combining~\Cref{protocol:reductionFromCKS} with~\Cref{protocol:secondBad} using the Balancing~\Cref{lemma:balancing}, we get an ROT protocol with security strictly better than the protocol discussed in~\Cref{subsub:badProtocolBalance} and is stated in the following theorem.

\begin{theorem}\label{thm:ROTWithCKSAndBalance} 
  There exists a quantum protocol for ROT with cheating advantage $\gamma = 1.239$. 
\end{theorem} 

We now describe a way to execute many ROT protocols in succession using a simpler protocol which exhibits the same cheating probabilities as derived previously for~\Cref{protocol:reductionFromCKS}.  

\subsection{A simple sequence of Rabin oblivious transfer protocols} 

We now propose another way of performing Rabin oblivious transfer based on the OT protocols described in~\cite{amiri2021imperfect}. 
We describe the protocol below. 

\begin{protocol}{A sequence of Rabin oblivious transfers.}\label{protocol:ROTFrom12ASR}
 
 \begin{itemize}
\item For each $i \in \{ 1, \ldots, N \}$, Bob picks a pair of bits $(x_0^i, x_1^i)$, all independently and uniformly at random. 
\item For each $i \in \{ 1, \ldots, N \}$, Bob prepares the two-qubit states $\ket{\Psi_{x_0^i x_1^i}}$, where we define: 
\begin{equation}
\ket{\Psi_{00}} =
\ket{00},\quad 
\ket{\Psi_{01}} = 
\ket{++}, \quad 
\ket{\Psi_{10}}  = 
\ket{--}, \quad 
\ket{\Psi_{11}} = 
\ket{11},
\quad 
\end{equation}
and $\ket{\pm} := (\ket{0}\pm\ket{1})/\sqrt{2}$. Bob sends the states to Alice.

\item Alice selects a subset $S$ of the set $\{ 1, \ldots, N \}$ of size $\left \lfloor \sqrt{N}  \right \rfloor$ uniformly at random.
For each $j \in S$, Alice asks Bob to announce the state $\ket{\Psi_{x_0^j, x_1^j}}$. 

\item If Bob announces $\ket{++}$ or $\ket{--}$, then Alice measures both qubits in the $\{ \ket{+}, \ket{-} \}$ basis, otherwise Alice measures in the $\{ \ket{0}, \ket{1} \}$ basis. 
Alice aborts if any of her outcomes are not consistent with Bob's purported states. 
If Alice does not abort, then the $\left \lfloor \sqrt{N}  \right \rfloor$ states used for these tests are discarded.

\item For each $j \notin S$, Alice measures the first qubit of $\ket{\Psi_{x_0^j x_1^j}}$ in the $\{ \ket{0}, \ket{1} \}$ basis and the second qubit of $\ket{\Psi_{x_0^j x_1^j}}$ in the $\{ \ket{+}, \ket{-} \}$ basis. 
This way Alice learns one of Bob's bits with certainty while gaining no knowledge of the other (each of her measurement outcomes allows her to rule out one of Bob's possible states). 
For example, if her measurement outcome is 
\begin{itemize} 
\item $00$: then we have $(x_0^j, x_1^j) = (0,0)$ or $(0,1)$, i.e., the first bit is $0$; 
\item $01$: then we have $(x_0^j, x_1^j) = (0,0)$ or $(1,0)$, i.e., the second bit is $0$; 
\item $10$: then we have $(x_0^j, x_1^j) = (0,1)$ or $(1,1)$, i.e., the second bit is $1$; 
\item $11$: then we have $(x_0^j, x_1^j) = (1,0)$ or $(1,1)$, i.e., the first bit is $1$. 
\end{itemize} 

\item For each $j \notin S$, Bob declares which of the two bits will be used for Rabin oblivious transfer. This is chosen uniformly at random for each $j$.  
    \end{itemize}
\end{protocol} 

As before, the completeness follows from the completeness of the underlying OT protocol. 
Again we have to carefully consider what a cheating Alice can do to further her cheating probability over what she can achieve if she uses the cheating strategy which is optimal in the OT protocol. 

The optimal cheating probabilities for Alice and Bob are given in the following lemma. 

\begin{lemma}\label{lemma:ROTFrom12ASR}
    The ROT~\Cref{protocol:ROTFrom12ASR} has cheating probabilities $P_A^* = \cos^2(\pi/8)$ for all $N$ and $P_B^* = 3/4$ in the limit $N \rightarrow\infty$ implying $\gamma_A = 1.138$ and $\gamma_B = 1.5$, resulting in a cheating advantage of $\gamma = 1.5$ in the limit.   
\end{lemma}

The full security analysis for the above protocol can be found in~\Cref{appendix:proofReductionToASR}. 

\subsection{A constant lower bound on $\gamma$ for any quantum protocol for Rabin oblivious transfer}\label{sub:lowerBoundROT}

The two main objectives for understanding the limitations of the security for any cryptographic task is to (a) find good protocols and the corresponding security measures, and (b) find lower bounds on such security measures. 
By improving either of these objectives brings us that much closer to finding the \emph{best} protocol. 

So far, we have been concentrating on objective (a), finding good quantum protocols. 
Now, we consider objective (b), finding lower bounds. 

\paragraph{Why care about lower bounds?}
Suppose one were to show a constant lower bound on $\gamma$, that is, suppose $\gamma \geq c$ for some constant $c > 1$. 
This implies two things. 
First, this means that obviously one cannot find protocols with $\gamma < c$, which could suggest that any existing protocols are better than expected (by being now closer to the optimal $\gamma$). 
Secondly, and importantly, this \emph{rules out security amplification attempts}. 
One may wonder if ROT can be repeated many times and then somehow by putting all the pieces together reduce the cheating probabilities to their ideal goals. 
Having a constant lower bound rules out such attempts, since near-perfect protocols cannot exist, even in the limit of a family of protocols. 
Thus, just by having a constant lower bound, regardless of the constant, provides insight towards how one may design optimal protocols. 
We therefore ask the question: 

\smallskip 
\begin{quote} 
\textit{Does there exist a constant lower bound for Rabin oblivious transfer?} 
\end{quote} 
\smallskip 

At first glance, it is not always obvious how to guess this. 
For example, quantum protocols for strong coin flipping (defined shortly) have a constant lower bound~\cite{kitaev2002quantum} while quantum protocols for weak coin flipping do not~\cite{mochon2007quantum}. 
What about Rabin oblivious transfer? 
In this section, we answer this question in the affirmative. 

\subsubsection{Strong coin flipping}\label{sub:SCF} 

We have discussed (unbalanced) weak coin flipping between Alice and Bob, the task of strong coin flipping is similar. 
Roughly, a strong coin flipping protocol is defined such that: 
\begin{itemize} 
\item The probability that the outcome is HEADS is $1/2$ when Alice and Bob are honest, and they both output the same coin flip; 
\item Alice can force the outcome HEADS with probability $P_{A,0}$;
\item Alice can force the outcome TAILS with probability $P_{A,1}$;
\item Bob can force the outcome HEADS with probability $P_{B,0}$;
\item Bob can force the outcome TAILS with probability $P_{B,1}$. 
\end{itemize} 
A formal definition is given in~\Cref{def:SCF}. 
While much is now known about the limits of the achievable values of $P_{A,0}$, $P_{A,1}$, $P_{B,0}$, $P_{B,1}$, in this discussion we only require one bound, proven by Kitaev~\cite{kitaev2002quantum} (see~\Cref{thm:KitaevLB}): 
\begin{equation}\label{eq:KitaevLowerBoundSCF} 
    P_{A,0} \cdot P_{B,0} \geq 1/2. 
\end{equation} 
In other words, in \emph{any} coin flipping protocol, either Alice or Bob can force the outcome $0$ with probability at least $1/\sqrt{2} > 1/2$. 

\subsubsection{A reduction from Rabin oblivious transfer to strong coin flipping}\label{subsub:redRabinOTToSCF} 

We present our reduction below. 

\begin{protocol}{Coin flipping from Rabin oblivious transfer}\label{protocol:RedRabinToCF}
    \begin{itemize} 
    \item Alice and Bob perform an ROT protocol with Bob's data bit denoted by $y$. 
    We denote Alice's ROT output by the trit $\hat{g} \in \{ 0, 1, NULL \}$. If Alice asserts the receipt of data, then $\hat{g} = y$, else $\hat{g} = NULL$. 
    \item Alice defines the assert bit $a \in \{ 0, 1 \}$ to be $0$ if $\hat{g} \in \{ 0, 1 \}$, and to be $1$ if $\hat{g} = NULL$.             
    \item Bob chooses $b$ uniformly at random and sends $b$ to Alice. 
    \item Alice sends $a$ and $\hat{g}$ to Bob.
            Bob checks if $(a,y,\hat{g})$ are consistent, i.e., he checks if $\hat{g} = y$ when $a = 0$.
                If $\hat{g} \neq y$, Bob aborts. 
    \item If Alice and Bob do not abort, then they both output HEADS if $a = b$, and TAILS otherwise. 
    \end{itemize} 
\end{protocol} 

We first remark that this protocol is complete, i.e., both Alice and Bob output the same coin flip whose value is generated uniformly at random. 
Below, we investigate the security of this coin flipping protocol with respect to the security of the ROT protocol. 

\begin{lemma}\label{lemma:LB}
In~\Cref{protocol:RedRabinToCF}, we have 
\begin{equation} 
P_{A,0} = \frac{1}{2} + \frac{1}{2} P_A^* 
\quad \text{ and } \quad 
P_{B,0} = P_B^*,  
\end{equation} 
where $P_A^*$ and $P_B^*$ are the cheating probabilities in the ROT subroutine. 
\end{lemma}  

The proof of the above lemma can be found in~\Cref{appendix:proofReductionToSCF}. 

By combining Equations~(\ref{eq:defCheatingAdvantage}) and~(\ref{eq:KitaevLowerBoundSCF}), i.e., that 
\begin{equation} 
\gamma = \max \{ \gamma_A, \gamma_B \}, \quad \gamma_A = \frac{P_A^*}{3/4}, \quad \gamma_B = \frac{P_B^*}{1/2}, \quad \text{ and } \quad P_{A,0} \cdot P_{B,0} \geq 1/2 
\end{equation} 
together with~\Cref{lemma:LB}, one gets the quadratic inequality 
\begin{equation} 
\frac{1}{2} \gamma \left( \frac{1}{2} + \frac{1}{2} \cdot \frac{3}{4} \gamma \right) \geq \frac{1}{2},
\end{equation} 
implying that $\gamma \geq \dfrac{2}{3} \left( \sqrt{7} - 1 \right) > 1$. 

Therefore, we have the following theorem. 

\begin{theorem} \label{thm:LB}
Any quantum protocol for Rabin oblivious transfer has a cheating advantage satisfying $\gamma \geq \dfrac{2}{3} \left( \sqrt{7} - 1 \right) \approx 1.097 > 1$. 
\end{theorem} 

By combining Theorems~\ref{thm:ROTWithCKSAndBalance} and~\ref{thm:LB}, we get the following corollary. 

\begin{corollary} 
    The optimal quantum protocol for Rabin oblivious transfer has cheating advantage $\gamma \in [1.097, 1.239]$. 
\end{corollary} 

\subsection*{Acknowledgments}
JS thanks Dominique Unruh for pointing out the OT to ROT reduction. AB thanks Or Sattath and Atul Mantri for useful discussions and pointing out the references for protocols under noisy communication. JW is supported by the National Research Foundation, Singapore and A*STAR under its Quantum Engineering Programme (NRF2021-QEP2-01-P06).
EA was supported by the UK Engineering and Physical Sciences Research Council (EPSRC) under
Grants No. EP/T001011/1 and EP/Z533208/1.

\bibliographystyle{alpha}
\bibliography{references} 

\appendix 

\section{Definition of other cryptographic tasks}\label{appendix:background} 

We now define a variant of oblivious transfer, namely 1-out-of-2 OT, and coin flipping with their underlying security notions.

\begin{definition}[1-out-of-2 OT]\label{def:1-out-of-2 OT}
1-out-of-2 OT is the cryptographic task between two parties (Alice and Bob) such that
\begin{itemize}
    \item Alice chooses $a \in \{0,1\}$ uniformly at random and Bob chooses $(x_0, x_1) \in \{0,1\}^2$ uniformly at random.
    \item If both Alice and Bob are honest, then neither Alice nor Bob aborts and Alice outputs $z = x_a$.
\end{itemize}
\end{definition}

We define the following notions of security for protocols that implement $1$-out-of-$2$ OT.
\begin{itemize}
    \item \emph{Completeness:} If both Alice and Bob are honest, Alice learns the bit $x_a$ unambiguously.
    
    \item \emph{Cheating Alice:} A dishonest Alice attempts to learn both $x_0$ and $x_1$. The cheating probability of Alice is given by
    \begin{equation*}
        P_A^{OT} = \max \Pr[\text{Alice correctly guesses both $x_0$ and $x_1$}],
    \end{equation*}
    where the maximum is taken over all cheating strategies of Alice. Note that $P_A^{OT} \geq 1/2$ since she can simply follow the protocol to learn $x_a$ perfectly (completeness) and can correctly guess $x_{\overline{a}}$ by making a guess uniformly at random.

    \item \emph{Cheating Bob:} A dishonest Bob attempts to learn the bit $a$. The cheating probability of Bob is given by
    \begin{equation*}
        P_B^{OT} = \max \Pr[\text{Bob correctly guesses $a$}],
    \end{equation*}
    where the maximum is taken over all cheating strategies of Bob. Note that $P_B^{OT} \geq 1/2$ since he can correctly guess Alice's bit $a$ by guessing one of the two values uniformly at random.
\end{itemize} 

We note that there are alternative definitions of this task where the inputs are fixed beforehand or possibly generated randomly within the protocol itself. Many of these definitions are equivalent to the one presented above and we refer the reader to~\cite{brassard87,Crepeau1988equivalence,chailloux2013lower,amiri2021imperfect} for proofs and discussions.

We next discuss coin flipping and provide its security definitions.   

\begin{definition}[Coin flipping]\label{def:SCF}
     Coin flipping is the cryptographic task between two parties (Alice and Bob) where they communicate to agree on a common binary outcome ($0$ or $1$). A weak version of this task is a variant where Alice \emph{wins} if the outcome is $1$ while Bob \emph{wins} if the outcome is $0$.
\end{definition}

We consider the following notions of security for a given weak coin flipping protocol.
\begin{itemize}
    \item \emph{Completeness:} If both Alice and Bob are honest, then neither party aborts and the shared outcome is generated uniformly at random. 
    \item \emph{Cheating Alice:} If Bob is honest, then Alice's cheating probability is defined with respect to the outcome $d \in \{0,1\}$ as
    \begin{equation*}
        P_{A,d} = \max \Pr[\text{Alice successfully forces Bob to accept the outcome $d$}],
    \end{equation*}
    where the maximum is taken over all possible cheating strategies of Alice. Note that here $P_{A,d} \geq 1/2$ since Alice can simply choose to follow the protocol honestly to observe the outcome $d$ uniformly at random.
    \item \emph{Cheating Bob:} If Alice is honest, then Bob's cheating probability is defined with respect to the outcome $d \in \{0,1\}$ as
    \begin{equation*}
        P_{B,d} = \max \Pr[\text{Bob successfully forces Alice to accept the outcome $d$}],
    \end{equation*}
    where the maximum is taken over all possible cheating strategies of Bob. As before, here $P_{B,d} \geq 1/2$ since Bob can simply choose to follow the protocol honestly to observe the outcome $d$ uniformly at random. 
\end{itemize}

The variant of the coin flipping task where Alice and Bob do not have any preferred choice of outcome is known as strong coin flipping.
We now state Kitaev's lower bound for strong coin flipping which is useful to deduce lower bounds for other two-party primitives.
\begin{prop}[Kitaev's lower bound~\cite{kitaev2002quantum}]\label{thm:KitaevLB}
    For all quantum coin flipping protocols, we have 
    \begin{equation}
        P_{A,d} P_{B,d} \geq 1/2
    \end{equation}
    where $d \in \{0,1\}$ and $P_{A,d}$ and $P_{B,d}$ are as defined above.
\end{prop}

\section{Security proof of~\Cref{protocol:reductionFromCKS}}\label{appendix:proofReductionToCKS} 

We now provide a proof of~\Cref{lemma:ROTWithCKS}. We reproduce the protocol below, for convenience.

\begin{itemize} 
        \item Alice chooses a bit $a \in \{0,1\}$ uniformly at random and creates the two-qutrit state
            \begin{equation*}
                \ket{\phi_a} = \frac{1}{\sqrt{2}}\ket{aa} + \frac{1}{\sqrt{2}}\ket{22} \in \mathcal{A} \otimes \mathcal{B}.
            \end{equation*}
        \item Alice sends the qutrit $\mathcal{B}$ to Bob.
        \item Bob chooses data bit $y \in \{0,1\}$ and $z \in \{0,1\}$  uniformly at random.
        \item Bob chooses a permutation bit $p \in \{0,1\}$ uniformly at random and assigns $(x_0,x_1) = (y,z)$ if $p = 0$, or $(x_0,x_1) = (z,y)$, otherwise.
        \item Bob applies the unitary
        \begin{equation*}
            U_{x_0x_1} = 
            \begin{bmatrix}
                (-1)^{x_0} & 0 & 0 \\
                0 & (-1)^{x_1} & 0 \\
                0 & 0 & 1
            \end{bmatrix} 
        \end{equation*}
        and sends back $\B$ to Alice.
        \item Alice determines the value of $x_a$ using the two-outcome measurement:
        \begin{equation*}
            \{\Pi_{0} \coloneqq \kb{\phi_a}_{\A \otimes \B}, \Pi_{1} \coloneqq \mathbb{1}_{\A \otimes \B} - \kb{\phi_a}  \}.
        \end{equation*}
        \item Bob sends $p$ to Alice.
        \item If $a = p$, Alice asserts the receipt of $y = x_a$, else she asserts NULL.
\end{itemize}

\begin{lemma}
    The optimal success probability ($P_A^*$) with which dishonest Alice can cheat in~\Cref{protocol:reductionFromCKS} is $\cos^2(\pi/8)$.
\end{lemma}

\begin{proof}
We have already argued in the main text that the most general strategy for Alice has her creating the state 
\begin{equation} 
\alpha \ket{0} + \beta \ket{1} + \gamma \ket{2}, 
\end{equation} 
where $\alpha, \beta, \gamma$ are nonnegative and satisfy $\alpha^2 + \beta^2 + \gamma^2 = 1$, 
and sending this qutrit to Bob. 
After Bob's unitary, Alice now has the state
\begin{equation} 
\ket{\psi_{x_0, x_1}} = \alpha (-1)^{x_0} \ket{0} + \beta (-1)^{x_1} \ket{1} + \gamma \ket{2}, 
\end{equation} 
which encodes his two bits. 
Since Alice wants to learn $y$, she should wait until $p$ is revealed which tells here whether she wants to learn $x_0$ or $x_1$ (that is, which bit of $(x_0, x_1)$ is really $y$). 
If $p=0$, that is, Alice wants to learn $x_0$, and she can do this with maximum probability given by 
\begin{equation} 
\frac{1}{2} + \frac{1}{4} \| \rho_0(0) - \rho_1(0) \|_1, 
\end{equation} 
where the norm is the trace norm and 
$\rho_0(0) = \frac{1}{2} \kb{\psi_{00}} + \frac{1}{2} \kb{\psi_{01}}$ is the encoding with $x_0 = 0$ and 
$\rho_1(0) = \frac{1}{2} \kb{\psi_{10}} + \frac{1}{2} \kb{\psi_{11}}$ is the encoding with $x_0 = 1$.  
A quick calculation shows that 
\begin{equation} 
\rho_0(0) - \rho_1(0) = \begin{pmatrix}
    0 & 0 & 2\alpha\gamma \\ 
    0 & 0 & 0 \\ 
    2\alpha\gamma & 0 & 0 
\end{pmatrix} 
\end{equation} 
which has at most two nonzero eigenvalues, $\pm 2 \alpha \gamma$. 
Therefore, Alice can cheat with maximum probability 
\begin{equation} 
\frac{1}{2} + \alpha \gamma  
\end{equation} 
in this case.
 
If $p=1$, that is, Alice wants to learn $x_1$, and she can do this with maximum probability given by 
\begin{equation} 
\frac{1}{2} + \frac{1}{4} \| \rho_0(1) - \rho_1(1) \|_1, 
\end{equation} 
where 
$\rho_0(1) = \frac{1}{2} \kb{\psi_{00}} + \frac{1}{2} \kb{\psi_{10}}$ is the encoding with $x_1 = 0$ and 
$\rho_1(1) = \frac{1}{2} \kb{\psi_{01}} + \frac{1}{2} \kb{\psi_{11}}$ is the encoding with $x_1 = 1$.  
A quick calculation shows that 
\begin{equation} 
\rho_0(1) - \rho_1(1) = \begin{pmatrix}
    0 & 0 & 0 \\ 
    0 & 0 & 2 \beta \gamma \\ 
    0 & 2 \beta \gamma & 0 
\end{pmatrix} 
\end{equation} 
which has at most two nonzero eigenvalues, $\pm 2 \beta \gamma$. 
Therefore, Alice can cheat with maximum probability 
\begin{equation} 
\frac{1}{2} + \beta \gamma 
\end{equation} 
in this case.

Since the choice of $p$ is uniformly random, Alice's overall cheating probability is given as 
\begin{equation} 
\frac{1}{2} + \frac{1}{2} \beta \gamma + \frac{1}{2} \alpha \gamma. 
\end{equation} 
A simple optimization over the parameters $\alpha, \beta, \gamma \geq 0$ satisfying $\alpha^2 + \beta^2 + \gamma^2 = 1$ yeilds that $\alpha = \beta = 1/2$ and $\gamma = 1/\sqrt{2}$ is the optimal solution from which we can calculate the optimal cheating probability as $\frac{1}{2} + \frac{1}{2 \sqrt{2}} = \cos^2(\pi/8)$, as desired. 

\end{proof}

\begin{lemma}
    The optimal success probability ($P_B^*$) with which dishonest Bob can cheat in~\Cref{protocol:reductionFromCKS} is $3/4$.
\end{lemma}

\begin{proof}
    It is necessary and sufficient for dishonest Bob in~\Cref{protocol:reductionFromCKS} to simply learn Alice's bit $a$ as Bob can then set his choice for $p$ accordingly, thereby successfully guessing Alice's assertion of receiving $y$ or NULL. 
    Bob can maximally distinguish the two states $\Tr_{\A}(\kb{\phi_0})$ and $\Tr_{\A}(\kb{\phi_1})$ in the Holevo-Helstrom basis resulting in $P_B^* = 3/4$.
\end{proof}
 
\section{Security proof of~\Cref{protocol:ROTFrom12ASR}}\label{appendix:proofReductionToASR} 

We next provide a proof of~\Cref{lemma:ROTFrom12ASR}.

\subsection{Cheating Alice} 

\begin{theorem}
    The optimal success probability ($P_A^*$) with which dishonest Alice can cheat in~\Cref{protocol:ROTFrom12ASR} is $\frac{1}{4}(2+\sqrt{2})$ for all $N$.
\end{theorem}

\begin{proof}
In this protocol Alice only receives states, and sends nothing to Bob. This means that the only cheating method available to Alice is to use a measurement which is different from what she would measure if honest. As Bob decides which of the two bits that counts for each of the remaining $N - \left \lfloor \sqrt{N}  \right \rfloor \approx N$ states, Alice's best strategy is to store these states until Bob declares which bit value counts. Alice then makes a minimum-error measurement individually on each state, aiming to learn the bit that Bob chose. Due to symmetry, Alice's probability of correctly obtaining the first bit is the same as for the second. When Bob declares that the first bit should be used, the two states Alice has to distinguish between are given by
\begin{equation}
\rho_0=\frac{1}{2}(\kb{00}+\kb{++}), \quad \rho_1=\frac{1}{2}(\kb{11}+\kb{--}).
\label{eq:AliceCheatingASR}
\end{equation}
Here Alice's best strategy is to make a minimum-error measurement on each instance individually. Since there are only two states to distinguish between, we use the standard result for the probability of success, the so-called Holevo-Helstrom measurement, yielding 
\begin{equation}
p_s=\frac{1}{2}(1+\Tr(|p_0\rho_0-p_1\rho_1|)),
\end{equation}
where $|A|=\sqrt{A^\dagger A}$. This gives Alice's cheating probability as
\begin{equation}
P_A^*=\frac{1}{4}(2+\sqrt{2})\approx0.853
\end{equation}
for all choices of $N$.
\end{proof}

\subsection{Cheating Bob}  

\begin{lemma}
    The optimal success probability ($P_B^*$) with which dishonest Bob can cheat in~\Cref{protocol:ROTFrom12ASR} is $\frac{3}{4}$ as $N\rightarrow\infty$.
\end{lemma}

\begin{proof}
The proof follows the same argument as before. A cheating Bob in the 1-out-of-2 protocol wants to know which bit Alice received. This is exactly what he needs to know whether or not Alice has received the bit value in the Rabin protocol we have constructed as well. Therefore, Bob’s cheating probability is 
\begin{equation}
P_B^*=\frac{3}{4},
\end{equation}
which follows directly from~\cite{amiri2021imperfect}, 
and is only valid in the limit as $N\rightarrow\infty$. This limit assures the fraction tested becomes negligible disallowing further cheating possibilities for Bob.
\end{proof} 

\section{Security proof of~\Cref{protocol:RedRabinToCF}}\label{appendix:proofReductionToSCF}

We now prove the security of the coin flipping protocol with respect to the security of the ROT subroutine. 
We reproduce the protocol below, for convenience. 

    \begin{itemize} 
    \item Alice and Bob perform an ROT protocol with Bob's data bit denoted by $y$. 
    We denote Alice's ROT output by the trit $\hat{g} \in \{ 0, 1, NULL \}$. 
    \item Alice defines the assert bit $a \in \{ 0, 1 \}$ to be $0$ if $\hat{g} \in \{ 0, 1 \}$, and to be $1$ if $\hat{g} = NULL$.             
    \item Bob chooses $b$ uniformly at random and sends $b$ to Alice. 
    \item Alice sends $a$ and $\hat{g}$ to Bob.
            Bob checks if $(a,y,\hat{g})$ are consistent, i.e., he checks if $\hat{g} = y$ when $a = 0$.
                If $\hat{g} \neq y$, Bob aborts. 
    \item If Alice and Bob do not abort, then they both output HEADS if $a = b$, and TAILS otherwise. 
    \end{itemize} 

\subsection{Cheating Alice} 

Since Alice wants the outcome HEADS, she must send back $a = b$ in the last message. 
However, this is only accepted if $(a, y, \hat{g})$ are consistent. 

If $b = 0$, then Alice must return $a = 0$, in which case Bob accepts the value of $a$ if $\hat{g} = y$. 
The maximum probability that Alice can pass the test in this case is the maximum probability that she can learn $y$ in the ROT protocol, which is precisely $P_A^*$. 

If $b = 1$, then Alice must return $a = 1$, but in this case, Bob is expecting $\hat{g} = NULL$, and thus if she sends this, Bob has nothing to check, and will simply accept the value of $a$. 

Since $b$ is chosen independently from the rest of the protocol, we have that $P_{A,0} = \frac{1}{2} + \frac{1}{2} P_A^*$, as desired. 

\subsection{Cheating Bob} 
Since Bob want HEADS, he must infer the value of $a$ ahead of time when he sends $b$. 
In other words, he must send $b = a$. 
Therefore, he is successful in forcing Alice to output HEADS with the same probability with which he can learn $a$ in the protocol. 
Since the value of $a$ corresponds to exactly whether or not Alice asserts the value of Bob's bit $y$, we have that the maximum probability with which Bob learns $a$ in the coin flipping protocol is simply $P_B^*$. 

\end{document}